\documentclass[11pt,fleqn]{article}
\usepackage[utf8]{inputenc}
\usepackage[T1]{fontenc}
\usepackage{lmodern}
\usepackage[a4paper,margin=2.7cm]{geometry}
\usepackage[authoryear]{natbib}
\providecommand{\bsize}{\small}
\usepackage{amsmath,amssymb,amsfonts,mathtools,amsthm}
\IfFileExists{bbm.sty}{\usepackage{bbm}}{\newcommand{\mathbbm}[1]{\mathbb{##1}}}
\usepackage{bm}
\usepackage{graphicx}
\usepackage{booktabs}
\usepackage{tabularx}
\usepackage{longtable}
\usepackage[section]{placeins}
\usepackage{nicefrac}
\usepackage{url}
\usepackage{xcolor}
\usepackage{microtype}
\IfFileExists{algorithm2e.sty}{
  \usepackage[ruled,vlined]{algorithm2e}
}{
  \usepackage{float}
  \floatstyle{ruled}
  \newfloat{algorithm}{tbp}{loa}
  \floatname{algorithm}{Algorithm}
}
\usepackage[hidelinks]{hyperref}

\definecolor{revblue}{RGB}{0,0,0}
\definecolor{warnorange}{RGB}{0,0,0}
\definecolor{critred}{RGB}{0,0,0}
\definecolor{postlassopurple}{RGB}{0,0,0}
\definecolor{newgreen}{RGB}{0,0,0}
\definecolor{outsourcegreen}{RGB}{0,0,0}
\newenvironment{outsourceblock}{}{}
\DeclareRobustCommand{\outsource}[1]{#1}
\newcommand{\outsourceflag}[1]{}

\DeclareRobustCommand{\rev}[1]{#1}
\DeclareRobustCommand{\plrev}[1]{#1}
\DeclareRobustCommand{\newrev}[1]{#1}
\newenvironment{newrevblock}{}{}

\newtheorem{theorem}{Theorem}[section]
\newenvironment{acknowledgement}{\section*{Acknowledgements}}{}

\AtBeginDocument{\let\cite\citep}
\graphicspath{{./}}
\title{Model-based bootstrap inference for Cox models after Lasso selection}
\author{Lena Schemet\textsuperscript{1,*}, Andreas Groll\textsuperscript{2},
and Sarah Friedrich-Welz\textsuperscript{1,3}\\[0.8em]
\small \textsuperscript{1}Mathematical Statistics and Artificial Intelligence in Medicine,\\
\small University of Augsburg, Universit\"atsstra\ss e 14, 86159 Augsburg, Germany\\
\small \textsuperscript{2}Department of Statistics, TU Dortmund University,
44221 Dortmund, Germany\\
\small \textsuperscript{3}Center for Advanced Analytics and Predictive Sciences (CAAPS),\\
\small University of Augsburg, Universit\"atsstra\ss e 14, 86159 Augsburg, Germany\\[0.5em]
\small \textsuperscript{*}Corresponding author: \texttt{lena.schemet@uni-a.de}}
\date{}

\begin{document}
\maketitle

\begin{abstract}
Inference after variable selection in Cox regression is difficult because simple Wald-type intervals after
selection can have poor finite-sample conditional coverage. \newrev{We study a model-based bootstrap for
inference after Cox--Lasso variable selection. The Cox--Lasso is fitted once to the original data to select a
set of variables, after which an unpenalized Cox model is fitted using only those variables. Bootstrap samples
are generated from a semiparametric plug-in Cox model specified by the coefficient estimate from this
unpenalized Cox refit, the Breslow baseline cumulative hazard estimator, and a plug-in censoring distribution.
In every bootstrap sample, the selected variable set is kept fixed and only the unpenalized Cox model is
refitted.} Under oracle-type sparse-model assumptions and standard Cox model regularity conditions, we prove
first-order bootstrap validity for \newrev{this procedure}. In the simulation scenarios considered, percentile
and studentized bootstrap intervals showed improved conditional coverage relative to the bootstrap-
Wald interval in several small- and moderate-sample settings. Their performance was broadly competitive
with debiased intervals, although the comparison depended on signal strength, tuning, and selection stability.
A SEER breast cancer example illustrates that the procedure can be implemented in a realistic survival
analysis and provides interpretable uncertainty quantification for effects reported after variable selection.
\end{abstract}

\medskip
\noindent\textbf{Keywords:} Bootstrap; Cox regression; High-dimensional
inference; Post-selection inference; Survival analysis.

\section{Introduction}
\label{sec:intro}

The Cox proportional hazards model
\cite{andersen1982cox,andersen1993model,cox1972regression}
is a central tool in time-to-event analysis and is widely used in biomedical
and epidemiological research, including cancer survival, cardiovascular risk,
and treatment-effect analyses in clinical trials
\cite{bradburn2003survival,kleinbaum2012survival,schober2018survival}.
Applications range from large epidemiological cohort studies
\cite{berrington2010bmi} to molecular survival studies in oncology
\cite{dave2004prediction}.

Time-to-event analysis is increasingly applied to high-dimensional biomedical
data, such as gene expression profiles and other omics measurements, where the
number of covariates may greatly exceed the sample size
\cite{bommert2022feature,bovelstad2007predicting}.
In such settings, Cox regression is frequently combined with regularization
methods for estimation and variable selection. In particular, Lasso
penalization is widely used for high-dimensional Cox modelling
\cite{hohberg2024flexible,huang2013lasso,simon2011regularization,tibshirani1997lasso},
including biomedical applications with genomic and other high-dimensional
predictors \cite{binder2008allowing,zhao2024omics_survival}.

While penalized Cox estimators are effective for model fitting and prediction,
statistical inference after variable selection remains challenging. Naive
Wald-type confidence intervals for selected coefficients may have distorted
coverage because the selection step alters the sampling distribution of the
estimator. This phenomenon is well known in post-selection inference
\cite{Berk2013PostSelection,lee2016exact,taylor2015statistical} and can be
particularly relevant in time-to-event analysis, where right censoring adds
variability to the partial-likelihood score and information structure
\cite{andersen1993model,fan2002variable}.

Existing approaches for inference after variable selection include sample
splitting \cite{Berk2013PostSelection}, debiased or desparsified estimators
\cite{vandegeer2014optimal,zhang2014confidence}, including Cox-specific work by
\citet*{xia2023coxdiverging}, and exact
conditional post-selection inference
\cite{lee2016exact,taylor2015statistical}. These methods provide important
theoretical guarantees, but they often target parameters defined independently
of the selected Cox model reported in practice or rely on strong structural
assumptions \cite{Berk2013PostSelection,taylor2015statistical,taylor2018post}.
Moreover, selective coverage guarantees do not necessarily imply satisfactory
marginal coverage in finite samples \cite{Kammer2022LassoSelective}.

Bootstrap methods offer a conceptually simple alternative for approximating
sampling distributions beyond first-order asymptotics
\cite{davison1997bootstrap,efron1993bootstrap}. For Cox regression, however,
bootstrap procedures must preserve the semiparametric model structure, the
risk-set dependence, and the censoring mechanism
\cite{andersen1993model,davison1997bootstrap,hjort1985}. In addition, combining
bootstrap procedures with $\ell_1$-penalized estimators is nontrivial because
the objective function is non-smooth and the selected model is data-dependent
\cite{chatterjee2011bootstrapping,knight2000asymptotics}.

\plrev{In this manuscript, we propose a model-based bootstrap approach for
post-selection inference in Cox regression. Building on the semiparametric
plug-in principle of Hjort~\cite{hjort1985}, we generate bootstrap samples from
\newrev{a semiparametric plug-in Cox model} rather than by ordinary case resampling. The procedure
therefore preserves the fitted survival model, the Breslow baseline hazard
structure, and the dependence induced by risk sets.}

\begin{newrevblock}
The proposed procedure consists of two stages. First, an
$\ell_1$-penalized Cox proportional hazards model, hereafter referred to as the
Cox--Lasso, is fitted to the original data to select a set of variables.
Second, an unpenalized Cox model is fitted using only these selected variables;
we refer to this second fit as the \emph{unpenalized Cox refit}. In the
bootstrap samples, the selected variable set is kept fixed and only the
unpenalized Cox model is refitted. The bootstrap distribution therefore
approximates the sampling distribution of the coefficient estimator from this
refitted model without incorporating additional variability from repeating
variable selection. This is the construction analyzed theoretically in this
manuscript. Our goal is not to provide exact selective inference or to replace
penalized Cox estimation, but to quantify uncertainty for coefficients reported
after Cox--Lasso variable selection. Accordingly, the resulting confidence
intervals refer to the unpenalized Cox refit with the originally selected
variable set held fixed.
\end{newrevblock}

\plrev{We formulate the method as a semiparametric plug-in bootstrap in sparse
Cox regression with growing model dimension and establish its first-order
validity under oracle-type sparse-model conditions. We then assess its
finite-sample behavior in simulation scenarios relevant to applied survival
analysis and compare percentile, Wald-type, and studentized bootstrap intervals
with established alternatives. The simulations are used to characterize the
settings in which the different interval constructions perform well or poorly,
without making a uniform ranking claim.}

\plrev{The remainder of the manuscript is organized as follows.
Section~\ref{sec:methods} introduces the statistical framework and the proposed
model-based bootstrap procedure. Section~\ref{sec:asymptotics} establishes its
asymptotic validity. Section~\ref{sec:simstudy} describes the simulation study
design. Section~\ref{sec:example} provides an illustrative data example, and
Section~\ref{sec:discussion} concludes.}

\section{Methods}
\label{sec:methods}

This section describes the methodological framework of the study.
We first introduce the model setting and the penalized estimation procedures,
then describe the model-based bootstrap used for post-selection uncertainty quantification.
The asymptotic justification is deferred to Section~\ref{sec:asymptotics}.

\subsection{Penalized Cox regression}
\label{sec:prelim}

We briefly introduce the statistical framework underlying the analysis,
including the Cox model under right censoring and the penalized
estimation procedure considered in this work.

For a vector $X\in\mathbb{R}^p$, let $X^{(j)}$ denote its
$j$-th component.
For $1\le q\le\infty$, let $\|X\|_q$ denote the usual $\ell_q$-norm,
in particular $\|X\|_\infty=\max_{1\le j\le p}|X^{(j)}|$.
For any vector or matrix $A$, let $A^\top$ denote its transpose.
We write $\mathbbm{1}(\cdot)$ for the indicator function.

%We use the notation
%\[
%x^{\otimes 0}=1,
%\qquad
%x^{\otimes 1}=x,
%\qquad
%x^{\otimes 2}=xx^\top.
%\]
Let $X=(X^{(1)},\ldots,X^{(p)})^\top\in\mathbb{R}^p$
denote a $p$-dimensional covariate vector and let $T$
denote the event time of interest.
Under right censoring with censoring time $C$, we observe
\[
W=T\wedge C,
\qquad
\delta=\mathbbm{1}(T\le C).
\]

For $n$ independent individuals, we observe
$(W_i,\delta_i,X_i)$ for $i=1,\ldots,n$, where
\[
W_i=T_i\wedge C_i.
\]
In counting-process notation, we define
\[
N_i(t)=\mathbbm{1}(W_i\le t,\delta_i=1),
\qquad
Y_i(t)=\mathbbm{1}(W_i\ge t),
\]
where $N_i(t)$ is the event counting process and $Y_i(t)$ the at-risk indicator.
We use uppercase notation for covariates throughout. In the asymptotic
analysis, all probability statements are conditional on the observed
covariates, so the design $\{X_i\}_{i=1}^n$ is treated as fixed.

We assume non-informative censoring in the sense of
\citet{andersen1993model}, that is, conditional on the covariates $X$,
the censoring mechanism is independent of the event process and does not
depend on the unknown regression parameters.

We consider the Cox proportional hazards model~\cite{cox1972regression},
\[
\lambda(t\mid X)=\lambda_0(t)\exp(X^\top\textcolor{revblue}{\beta^\circ}),
\qquad t\in[0,\tau],
\]
\rev{Here $\tau<\infty$ denotes the fixed study horizon or maximal follow-up time used in the theoretical arguments.}
where $\lambda_0(t)$ is an unknown baseline hazard function and
\rev{$\beta^\circ=(\beta^\circ_1,\ldots,\beta^\circ_p)^\top$}
denotes the true regression parameter vector.
The corresponding conditional survival function is
\[
S_T(t\mid X)
=
\mathbb{P}(T>t\mid X)
=
\exp\!\left\{-\Lambda_0(t)\exp(X^\top\textcolor{revblue}{\beta^\circ})\right\},
\qquad
\Lambda_0(t)=\int_0^t \lambda_0(u)\,du.
\]

Given the observed data $(W_i,\delta_i,X_i)$, the log partial likelihood
for a generic $\beta\in\mathbb{R}^p$ is
\begin{equation}
\label{eq:partiallik}
\ell(\beta)
=
\sum_{i=1}^n \delta_i
\left(
X_i^\top\beta
-
\log\sum_{j:W_j\ge W_i}\exp(X_j^\top\beta)
\right).
\end{equation}

In the absence of penalization, estimation of $\beta^\circ$
is based on maximizing the log partial likelihood $\ell(\beta)$.

\iffalse
For $k=0,1,2$, define the population quantities
\[
e^{(k)}(\beta,t)
=
E\!\left[
Y(t)X^{\otimes k}\exp(X^\top\beta)
\right].
\]
Set
\[
v(\beta,t)
=
\frac{e^{(2)}(\beta,t)}{e^{(0)}(\beta,t)}
-
\left(
\frac{e^{(1)}(\beta,t)}{e^{(0)}(\beta,t)}
\right)^{\otimes 2}.
\]
The population information matrix at $\beta^\circ$ is
\[
I(\beta^\circ)
=
\int_0^\tau
v(\beta^\circ,t)\,
e^{(0)}(\beta^\circ,t)\,
d\Lambda_0(t).
\]

Let
\[
U_n(\beta)
=
\sum_{i=1}^n
\int_0^\tau
\{X_i-\bar X(\beta,t)\}\,dN_i(t)
\]
denote the Cox score, and define the asymptotic covariance matrix of the score by
\[
\Sigma(\beta^\circ)
=
\lim_{n\to\infty}
\mathrm{Var}\!\left(
n^{-1/2}U_n(\beta^\circ)
\right).
\]
Under correct model specification,
\[
\Sigma(\beta^\circ)=I(\beta^\circ),
\]
see \cite[Ch.~VII]{andersen1993model}.
\fi 
We allow the dimension \(p=p_n\) to increase with \(n\).
To obtain sparse estimators in this regime, we consider the
\(\ell_1\)-penalized Cox estimator
\[
\hat\beta
=
\arg\max_{\beta\in\mathbb{R}^p}
\left\{
\ell(\beta)-\gamma_n\|\beta\|_1
\right\},
\]
where \(\gamma_n>0\) is the penalty level, i.e.\ the scalar multiplying
the \(\ell_1\)-penalty in the objective function.
\rev{In the high-dimensional asymptotic regime, \(\gamma_n\) may depend on
\(n\) and is typically chosen to dominate the stochastic fluctuations of the
Cox score. We assume the standard Cox--Lasso order
\[
\gamma_n \asymp \sqrt{\frac{\log p_n}{n}},
\]
where \(a_n \asymp b_n\) means that \(a_n/b_n\) is bounded away from zero and
infinity by positive constants
\cite{buhlmann2011statistics,huang2013lasso}.}
The \(\ell_1\)-penalty induces sparsity by shrinking some coefficients
exactly to zero and thereby performs variable selection. However, the Lasso
estimator is biased toward zero and may include spurious variables,
particularly in finite samples \cite{Wainwright2009Sharp}.

Let
 ${M^\circ}
=
\{j:{\beta^\circ_j}\neq 0\}$
denote the true active set, and let
$s^\circ
=
|M^\circ|
$
denote its cardinality. Throughout, we consider a high-dimensional sparse
regime in which \(p=p_n\) may diverge with \(n\), while the true active set is
small relative to the sample size, \(s^\circ\ll n\).
\rev{Because the full coefficient vector cannot be estimated without
additional structure when \(p\) may exceed \(n\), we impose a compatibility
condition that ensures sufficient information in sparse directions. The
condition is stated precisely in Section~\ref{sec:boot-highdim}.}

\newrev{We now describe the bootstrap procedure used to quantify uncertainty
for the coefficients from the unpenalized Cox refit based on the selected
variable set.}

\subsection{Model-based bootstrap after Cox--Lasso variable selection}
\label{sec:mbboot}

The procedure separates variable selection from subsequent uncertainty
quantification. Let \(\widehat\beta^L\) denote the Cox--Lasso estimator fitted
in the original data and
\(\widehat M=\{j:\widehat\beta^L_j\neq 0\}\) \newrev{the selected variable set}. The penalized
estimator is used only for \newrev{variable selection}, not as the inferential target, because
the \(\ell_1\)-penalty shrinks nonzero coefficients and introduces
regularization bias.

\begin{newrevblock}
After variable selection, an ordinary unpenalized Cox model is fitted using
only the variables in \(\widehat M\), yielding the unpenalized Cox refit
estimator
\[
\widetilde\beta_{\widehat M}\in\mathbb R^{|\widehat M|}.
\]
For theoretical comparisons with the original \(p\)-dimensional parameter
vector, we also regard \(\widetilde\beta_{\widehat M}\) as an element of
\(\mathbb R^p\) by setting its coordinates outside \(\widehat M\) to zero.
This zero-padding convention is used only for notation; the Cox refit and the
generation of bootstrap samples use only the variables in \(\widehat M\).
The resulting coefficient estimates are the inferential objects considered
here.

The bootstrap approximates the sampling distribution of the unpenalized Cox
refit with the variable set selected in the original data held fixed.
Specifically, \(\widehat M\) is selected once and reused in all bootstrap
samples without repeating the variable-selection step.
\end{newrevblock}

Following the semiparametric plug-in bootstrap principle of
\citet{hjort1985}, bootstrap data are generated conditionally on the observed
covariates from \newrev{the plug-in Cox model defined by the unpenalized Cox refit}, with linear predictor
\(\tilde\eta_i
 =X_{i,\widehat M}^{\top}\tilde\beta_{\widehat M}\). Here, $X_{i,\widehat M}$ denotes the subvector of Xi restricted to selected coordinates.
The baseline cumulative hazard is estimated by the Breslow estimator evaluated at

The baseline cumulative hazard is estimated by the Breslow estimator evaluated
at the post-Lasso refit. In counting-process notation,
\[
\hat\Lambda_0(t;\tilde\beta_{\widehat M})
=
\int_0^t
\frac{dN(u)}
{E_n^{(0)}(\tilde\beta_{\widehat M},u)},
\qquad
E_n^{(0)}(\tilde\beta_{\widehat M},t)
=
\sum_{j=1}^n
Y_j(t)\exp\!\bigl(X_{j,\widehat M}^{\top}\tilde\beta_{\widehat M}\bigr),
\]
where \(N(t)=\sum_{i=1}^nN_i(t)\) and
\(Y_i(t)=\mathbbm{1}(W_i\ge t)\)
\citep{andersen1982cox,andersen1993model}. Equivalently, if \(t_{(k)}\)
denotes the ordered distinct event times, \(d_k\) the number of events at
\(t_{(k)}\), and \(R(t)\) the risk set, then
\[
\hat\Lambda_0(t;\tilde\beta_{\widehat M})
=
\sum_{k:t_{(k)}\le t}
\frac{d_k}
{\sum_{j\in R(t_{(k)})}
\exp\!\bigl(X_{j,\widehat M}^{\top}\tilde\beta_{\widehat M}\bigr)}.
\]
This is the usual Breslow estimator applied to \newrev{the unpenalized Cox refit};
see also \citet{lin2007breslow}.

The fitted conditional survival function used to generate bootstrap event times
is
\[
\hat S_T(t\mid X_i)
=
\exp\!\left\{
-\hat\Lambda_0(t;\tilde\beta_{\widehat M})
\exp\!\bigl(X_{i,\widehat M}^{\top}\tilde\beta_{\widehat M}\bigr)
\right\}.
\]

Bootstrap event times are generated by inverse transformation. Specifically,
for independent \(U_i\sim\mathrm{Unif}(0,1)\),
\[
T_i^*
=
\hat\Lambda_0^{-1}\!\left(
\frac{-\log U_i}
{\exp\!\bigl(X_{i,\widehat M}^{\top}\tilde\beta_{\widehat M}\bigr)}
\right),
\]
where \(\hat\Lambda_0^{-1}\) denotes the generalized inverse of the estimated
stepwise baseline cumulative hazard.

To preserve the censoring structure while maintaining conditional independence
between event and censoring times, bootstrap censoring times are generated from
a plug-in estimator of the censoring distribution. Under
covariate-independent censoring, we use the Kaplan--Meier estimator fitted to
\((W_i,1-\delta_i)\), denoted by \(\hat F_C\) \citep{kaplan1958}. If censoring
depends on covariates, a conditional plug-in estimator
\(\hat F_C(\cdot\mid X_i)\) may be used instead. The theoretical analysis covers
both settings. The resulting bootstrap observations are
\(W_i^*=T_i^*\wedge C_i^*\) and
\(\delta_i^*=\mathbbm{1}\{T_i^*\le C_i^*\}\).

\newrev{In each bootstrap sample, the unpenalized Cox model is refitted using the
originally selected variable set \(\widehat M\), without repeating the selection
step, yielding
\(\tilde\beta_{\widehat M}^{*(b)}\), \(b=1,\ldots,B\). Confidence intervals for
the selected coefficients are then constructed from the bootstrap distribution
of
\(\tilde\beta_{\widehat M}^{*(b)}-\tilde\beta_{\widehat M}\).}

\begin{algorithm}[t]
\caption{Model-based bootstrap after Cox--Lasso variable selection}
\label{alg:coxplugin}
\begin{enumerate}
\item Fit the Cox--Lasso in the original data and obtain the \newrev{selected variable set}
\(\widehat M=\{j:\widehat\beta_j^L\neq 0\}\).

\item Refit the unpenalized Cox model on \(\widehat M\) to obtain the
\newrev{unpenalized Cox refit} \(\tilde\beta_{\widehat M}\).

\item Estimate the baseline cumulative hazard
\(\hat\Lambda_0(\cdot;\tilde\beta_{\widehat M})\) by the Breslow estimator,
and estimate the censoring distribution by \(\hat F_C\) or
\(\hat F_C(\cdot\mid X_i)\).

\item For \(b=1,\ldots,B\):
\begin{enumerate}
\item Draw bootstrap event times \(T_i^{*(b)}\) from the fitted Cox survival
law based on \((\tilde\beta_{\widehat M},\hat\Lambda_0)\), and draw bootstrap
censoring times \(C_i^{*(b)}\) from the censoring plug-in.

\item Form \(W_i^{*(b)}=T_i^{*(b)}\wedge C_i^{*(b)}\) and
\(\delta_i^{*(b)}
 =\mathbbm{1}\{T_i^{*(b)}\le C_i^{*(b)}\}\).

\item Refit the unpenalized Cox model using the originally \newrev{selected variable set}
\(\widehat M\) using the bootstrap sample, yielding
\(\tilde\beta_{\widehat M}^{*(b)}\).
\end{enumerate}

\item Construct confidence intervals for the selected coefficients from
\(\tilde\beta_{\widehat M}^{*(1)},\ldots,
\tilde\beta_{\widehat M}^{*(B)}\).
\end{enumerate}
\end{algorithm}

\newrev{Algorithm~\ref{alg:coxplugin} therefore describes the proposed
model-based post-selection bootstrap. The Cox--Lasso determines the reporting
variable set, the unpenalized Cox refit defines the coefficient vector used for inference,
and the bootstrap approximates the variability of this refit under the corresponding
plug-in Cox model. Repeating variable selection inside the bootstrap is not part of this
core procedure.}

\subsection{Confidence intervals}
\label{sec:confidence-intervals}

For each selected variable \(j\in\widehat M\), inference is based on the
\newrev{unpenalized Cox refit coefficient} \(\tilde\beta_{\widehat M,j}\), not on the
penalized Cox--Lasso coefficient \(\widehat\beta^L_j\). In the model-based
bootstrap, the corresponding bootstrap coefficient is
\(\tilde\beta_{\widehat M,j}^{*(b)}\), computed by refitting the unpenalized Cox
model on the same support \(\widehat M\) in bootstrap sample \(b\).

We consider three interval constructions
\citep{davison1997bootstrap,efron1993bootstrap,hall1992bootstrap}.

Let \(\widehat q_{j,\alpha}\) denote the empirical \(\alpha\)-quantile of
\(\{\tilde\beta_{\widehat M,j}^{*(1)},\ldots,
\tilde\beta_{\widehat M,j}^{*(B)}\}\). The percentile interval is
\(\mathrm{CI}^{\mathrm{perc}}_j(1-\alpha)
=[\widehat q_{j,\alpha/2},\widehat q_{j,1-\alpha/2}]\).

The Wald-type bootstrap interval uses the empirical standard deviation of the
bootstrap refits,
\[
\widehat\mu^{*}_{j,B}
=
\frac{1}{B}\sum_{b=1}^B
\tilde\beta_{\widehat M,j}^{*(b)},
\qquad
\widehat{\mathrm{se}}_{j,\mathrm{boot}}
=
\left\{
\frac{1}{B-1}
\sum_{b=1}^B
\left(
\tilde\beta_{\widehat M,j}^{*(b)}
-
\widehat\mu^{*}_{j,B}
\right)^2
\right\}^{1/2}.
\]
With \(z_{1-\alpha/2}\) denoting the standard normal quantile,
\[
\mathrm{CI}^{\mathrm{wald}}_j(1-\alpha)
=
\left[
\tilde\beta_{\widehat M,j}
-
z_{1-\alpha/2}\widehat{\mathrm{se}}_{j,\mathrm{boot}},
\;
\tilde\beta_{\widehat M,j}
+
z_{1-\alpha/2}\widehat{\mathrm{se}}_{j,\mathrm{boot}}
\right].
\]

The studentized interval is based on
\[
T^{*(b)}_j
=
\frac{
\tilde\beta_{\widehat M,j}^{*(b)}
-
\tilde\beta_{\widehat M,j}
}{
\widehat{\mathrm{se}}_{\widehat M,j}^{*(b)}
},
\]
where \(\widehat{\mathrm{se}}_{\widehat M,j}^{*(b)}\) is the Cox
information-based standard error computed in bootstrap sample \(b\) on the
originally selected support \(\widehat M\). Let
\(\widehat t_{j,\alpha}\) denote the empirical \(\alpha\)-quantile of
\(\{T_j^{*(1)},\ldots,T_j^{*(B)}\}\), and let
\(\widehat{\mathrm{se}}_{\widehat M,j}\) be the corresponding standard error
in the original refit. The studentized interval is
\[
\mathrm{CI}^{\mathrm{stud}}_j(1-\alpha)
=
\left[
\tilde\beta_{\widehat M,j}
-
\widehat t_{j,1-\alpha/2}
\widehat{\mathrm{se}}_{\widehat M,j},
\;
\tilde\beta_{\widehat M,j}
-
\widehat t_{j,\alpha/2}
\widehat{\mathrm{se}}_{\widehat M,j}
\right].
\]

%All intervals are therefore constructed for coordinates of the same selected-support Cox refit \(\tilde\beta_{\widehat M}\). The Cox--Lasso coefficient \(\widehat\beta^L\) enters only through the selected support \(\widehat M\).

\begin{newrevblock}
All three interval constructions target coordinates of the same unpenalized
Cox refit \(\widetilde\beta_{\widehat M}\). The Cox--Lasso estimator
\(\widehat\beta^L\) enters the procedure only through the selected variable set
\(\widehat M\).

The first-order validity of the studentized interval additionally requires
consistency of the Cox information-based standard-error estimators in the
original data and in the bootstrap samples. This consistency does not follow
from variable selection alone. Under the oracle-type selection-stability
assumptions imposed in Section~\ref{sec:asymptotics}, including suitable
conditions on the penalty level and signal strength, the event
\(\widehat M=M^\circ\) has probability tending to one. On this event, the
unpenalized Cox refit is the ordinary Cox partial-likelihood estimator in the
true active model. The required variance-consistency result then follows from
the usual consistency of the observed-information estimator for the ordinary
Cox model
\citep{andersen1982cox,andersen1993model,huang2013lasso}.
\end{newrevblock}

\section{Theoretical justification}
\label{sec:asymptotics}

This section provides a first-order theoretical justification for the
proposed model-based bootstrap procedure in Cox regression.
We first establish bootstrap validity in the classical fixed-dimensional
Cox model and then extend the result to the sparse
\(\ell_1\)-penalized Cox--Lasso setting using an oracle reduction argument.

\begin{newrevblock}
For the sparse Cox--Lasso setting, our results show that the conditional
distribution of
$\sqrt{n}\bigl(
\widetilde\beta_{\widehat M}^{*}
-
\widetilde\beta_{\widehat M}
\bigr)$
consistently approximates the first-order distribution of
$\sqrt{n}\bigl(
\widetilde\beta_{\widehat M}
-
\beta^\circ
\bigr).$

Here, \(\widetilde\beta_{\widehat M}\) is the unpenalized Cox refit estimator
based on the variable set \(\widehat M\) selected from the original data. The
same selected variable set is used in every bootstrap sample; variable
selection is not repeated within the bootstrap.

Percentile-type intervals follow most directly from this bootstrap
distributional approximation. Wald-type and studentized intervals additionally
require consistency of the quantities used to estimate or standardize their
sampling variability. The additional variance-consistency requirement is
stated explicitly for the studentized interval. Higher-order refinements are
not considered.
\end{newrevblock}

Throughout, \(\to_p\) denotes convergence in probability and
\(\Rightarrow\) weak convergence in distribution.
For bootstrap quantities,
\(\Rightarrow_P\) denotes conditional weak convergence in probability,
i.e.\ weak convergence of the bootstrap law conditional on the data.
For a random vector \(Z_n\) and a \(\sigma\)-field \(\mathcal D_n\),
\(\mathcal L(Z_n \mid \mathcal D_n)\) denotes the conditional distribution
(law) of \(Z_n\) given \(\mathcal D_n\).
The endpoint \(\tau\) is the fixed finite horizon introduced in
Section~\ref{sec:prelim}. We use the standard counting-process notation for
Cox models and for their bootstrap analogues, following the Andersen--Gill
formulation and the martingale framework for Cox regression
\citep{andersen1982cox,andersen1993model}. Full formulas for the risk-set
sums, score functions, observed information matrices, population information,
score covariance, and bootstrap counterparts are given in Supporting
Information~A. Related martingale and multiplier-bootstrap notation for
survival processes is discussed in work on wild and weird bootstrap procedures
for counting-process statistics
\citep{dietrich2023wild,dobler2015weird,dobler2018wild}.

\subsection{Model-based bootstrap validity in the classical Cox model}
\label{sec:boot-lowdim}

We begin with the case of fixed dimension $p$.
All regularity assumptions ensuring the counting-process formulation,
martingale structure, and asymptotic linearity of the Cox estimator
(Conditions~(C1)--(C6))
are stated in the Supporting Information.

Bootstrap validity additionally requires a consistent plug-in estimator
of the censoring distribution used to generate bootstrap censoring times.

\outsourceflag{detailed censoring plug-in conditions}
\begin{outsourceblock}
Depending on the censoring structure, we impose either:
\begin{enumerate}
\item[(C7)] \label{C7}
(Covariate-independent censoring (marginal KM plug-in))
$C\perp X$, bootstrap censoring times are drawn
$C_i^*\stackrel{iid}{\sim}\hat F_C$, and
\[
\sup_{t\in[0,\tau]}|\hat F_C(t)-F_C(t)|\to_p 0.
\]

\item[(C8)] \label{C8}
(Conditional censoring plug-in)
Bootstrap censoring times are drawn
$C_i^*\sim \hat F_C(\cdot\mid X_i)$, where
\[
\sup_{t\in[0,\tau],\, i\le n}
\big|\hat F_C(t\mid X_i)-F_C(t\mid X_i)\big|
\to_p 0.
\]
\end{enumerate}

\noindent
Condition~(C7) corresponds to covariate-independent censoring ($C\perp X$),
in which case $\hat F_C$ can be taken as the Kaplan--Meier estimator based on
$(W_i,1-\delta_i)$.
Condition~(C8) allows covariate-dependent censoring and requires a consistent
estimator of the conditional censoring law, e.g.\ via stratification, a Cox
model for censoring, or other flexible conditional survival estimators.
\end{outsourceblock}

\begin{theorem}[Bootstrap validity, fixed dimension]
\label{thm:app_mbboot_cox}
Assume the Cox model with fixed $p$ on $[0,\tau]$ and Conditions~(C1)--(C6).
Assume additionally either (C7) or (C8).

Generate bootstrap data conditionally on the observed sample
$\mathcal D_n=\{(W_i,\delta_i,X_i)\}_{i=1}^n$
by drawing independently, for $i=1,\dots,n$,
\[
T_i^* \mid X_i \sim \hat S(\,\cdot\mid X_i),
\qquad
\hat S(t\mid X_i)
=
\exp\{-\hat\Lambda_0(t)\exp(X_i^\top \hat\beta)\},
\qquad
C_i^* \sim \hat F_C(\cdot\mid X_i),
\]
independently of each other, and define
\[
W_i^* = T_i^* \wedge C_i^*,
\qquad
\delta_i^*=\mathbbm{1}(T_i^*\le C_i^*).
\]
Equivalently, $T_i^*$ may be generated by inverse transform sampling via
\[
T_i^*
=
\hat\Lambda_0^{-1}\!\left(
-\frac{\log U_i}{\exp(X_i^\top\hat\beta)}
\right),
\qquad
U_i\sim \mathrm{Unif}(0,1).
\]
Let $\hat\beta^*$ denote the Cox partial likelihood estimator
computed from $(W_i^*,\delta_i^*,X_i)$.

Then
\[
\mathcal L\!\left(
\sqrt{n}(\hat\beta^*-\hat\beta)
\mid \mathcal D_n
\right)
\Rightarrow_P
\mathcal L\!\left(
\sqrt{n}(\hat\beta-\beta^\circ)
\right).
\]
\end{theorem}

\begin{proof}[Proof sketch]
A detailed proof is provided in Supporting Information~A.3.3.

\smallskip
\noindent
\emph{(i) Plug-in correctness.}
Under Conditions~(C1)--(C6), the Cox estimator is consistent
and the Breslow estimator is uniformly consistent. Under either (C7) or (C8),
the censoring distribution used in the bootstrap generator converges uniformly
to the true censoring law. Hence the bootstrap data-generating mechanism
converges in probability to the true data-generating Cox model.

\smallskip
\noindent
\emph{(ii) Martingale representation.}
Both the original sample and the bootstrap sample admit multiplicative
intensity representations for the counting processes. The Cox score process can
therefore be written as a sum of martingale stochastic integrals.

\smallskip
\noindent
\emph{(iii) Linear expansion.}
A Taylor expansion of the score equation yields
\[
\sqrt n(\hat\beta-\beta^\circ)
=
I(\beta^\circ)^{-1} n^{-1/2} U_n(\beta^\circ)
+ o_p(1),
\]
and analogously in the bootstrap world,
\[
\sqrt n(\hat\beta^*-\hat\beta)
=
I(\hat\beta)^{-1} n^{-1/2} U_n^*(\hat\beta)
+ o_{P^*}(1),
\]
where $o_{P^*}(1)$ denotes a remainder term that converges to zero in
bootstrap probability conditional on the data.

\smallskip
\noindent
\emph{(iv) Conditional martingale CLT.}
The predictable quadratic variation of the bootstrap score converges in
probability to the same limit as in the original sample. A conditional
martingale central limit theorem for triangular arrays yields
\[
n^{-1/2} U_n^*(\hat\beta)
\Rightarrow_P
N(0,I(\beta^\circ)).
\]
Combining these steps with conditional Slutsky's theorem gives the result.
\end{proof}

\subsection{Extension to the sparse Cox--Lasso setting}
\label{sec:boot-highdim}

We now allow the dimension \(p=p_n\) to grow with \(n\) and use the Cox--Lasso
for variable selection. The raw penalized Cox--Lasso coefficient is not the
inferential target in the theorem below because the penalty induces shrinkage
bias. \newrev{The target is the coefficient from the unpenalized Cox refit defined in
Section~\ref{sec:mbboot}, with the selected variable set held fixed. This is the
main procedure covered by the theorem below.}

\outsourceflag{detailed sparse-regime assumptions}
\begin{outsourceblock}
In addition to Conditions~(C1)--(C5) and either (C7) or (C8), we impose
sparse Cox--Lasso assumptions. The penalty level satisfies
\[
\gamma_n\asymp \sqrt{\frac{\log p}{n}},
\qquad
\max_{1\le j\le p}
\left|
\frac1n U_{n,j}(\beta^\circ)
\right|
\le \gamma_n/2
\quad\text{with probability tending to one},
\]
the true support \(M^\circ=\{j:\beta^\circ_j\neq0\}\) is sparse, and the
population Cox information operator satisfies a restricted compatibility
condition on sparse cones. We also assume beta-min and original-sample
selection stability,
\[
\min_{j\in M^\circ}|\beta^\circ_j|
\ge c_{\min}\gamma_n,
\qquad
\mathbb P(\widehat M=M^\circ)\to1.
\]
These assumptions are standard in sparse Lasso theory and Cox--Lasso oracle
arguments \citep{buhlmann2011statistics,huang2013lasso,Wainwright2009Sharp}.
The full technical formulation, including the active-model dimension
requirement used for fixed-coordinate statements, is given in Supporting
Information~A.
\end{outsourceblock}

\begin{theorem}[\newrev{Model-based bootstrap validity with the selected
variable set held fixed}]
\label{thm:cox_lasso_bootstrap_app}

Assume the Cox model on \([0,\tau]\) and the oracle-type sparse-regime
conditions described above. Let \(\widehat M\) be the variable set selected
from the original data, and let
\(\widetilde\beta_{\widehat M}\in\mathbb R^{|\widehat M|}\) denote the
coefficient vector from the unpenalized Cox refit using the variables in
\(\widehat M\). For the theoretical statement below, this vector is regarded
as an element of \(\mathbb R^p\) by setting its coordinates outside
\(\widehat M\) to zero.

Generate a model-based bootstrap sample from the plug-in Cox law based on
\[
\bigl(
\widetilde\beta_{\widehat M},
\widehat\Lambda_0,
\widehat F_C
\bigr).
\]
In each bootstrap sample, refit the unpenalized Cox model using the same
selected variable set \(\widehat M\), yielding the bootstrap coefficient
vector \(\widetilde\beta_{\widehat M}^{*}\), with the same zero-padding
convention. Then, for fixed active coordinates or fixed finite-dimensional
active contrasts,
\[
\mathcal L\!\left(
\sqrt n\bigl(
\widetilde\beta_{\widehat M}^{*}
-
\widetilde\beta_{\widehat M}
\bigr)
\mid \mathcal D_n
\right)
\Rightarrow_P
\mathcal L\!\left(
\sqrt n\bigl(
\widetilde\beta_{\widehat M}
-
\beta^\circ
\bigr)
\right).
\]
\end{theorem}

\begin{proof}[Proof sketch]
A detailed proof is provided in Supporting Information~A.4. Let
$A_n=\{\widehat M=M^\circ\}$ denote the oracle selection event. By the
selection-stability assumption, $\mathbb P(A_n)\to1$. On $A_n$,
$\tilde\beta_{\widehat M}$ is the ordinary Cox estimator in the true active
model, and the bootstrap refits the same active model. The fixed-dimensional
bootstrap result therefore applies on $A_n$. Since replacing the random
selected support by the oracle support changes the law only by $o_P(1)$,
conditional Slutsky's theorem gives the result.
\end{proof}

For studentized intervals, the Cox information-based standard errors must be
consistent in the \newrev{original unpenalized Cox refit} and in the bootstrap refits on
the same support. On the oracle event this is the usual observed-information
consistency for the ordinary Cox partial-likelihood estimator on the active
model \citep[Ch.~VII]{andersen1993model}. Together with
Theorem~\ref{thm:cox_lasso_bootstrap_app}, this gives first-order validity of
the studentized intervals for fixed active coordinates.

\section{Simulation study}
\label{sec:simstudy}

This section investigates the finite-sample performance of the proposed
model-based bootstrap procedure for post-selection inference in Cox
regression. The study evaluates coverage accuracy, interval length,
and error rates under controlled data-generating mechanisms that vary
sample size, covariate dimension, correlation structure, baseline
hazard, and censoring proportion.

\subsection{Design}
\label{sec:design}
The simulation design follows the ADEMP framework \citep*{morris2019using},
which organizes simulation studies in terms of their aims, data-generating
mechanisms, estimands, methods, and performance measures. This structure is
used here to separate the scientific question from the technical choices made
in the simulation implementation.

\subsubsection{Objectives}

The primary objective of the simulation study was to assess whether the
proposed bootstrap procedure improves selective coverage after Lasso-based
variable selection compared with standard post-selection Wald-type inference.
As a secondary performance measure, we considered the length of the resulting
confidence intervals.

\subsubsection{Data-generating mechanisms}

The simulation study was conducted in two stages. A broader initial scenario
grid was used for exploratory analyses, whereas the primary analyses reported
in this manuscript focused on small and moderate sample sizes most relevant
for the proposed bootstrap procedure, with
\[
n\in\{30,40,60,75,80,100,125,150,175,200,250\}
\quad\text{and}\quad
p\in\{10,20\}.
\]
Larger and additional exploratory scenarios are summarized in the
~B.1.

Covariates were generated from a multivariate normal distribution with mean
zero and autoregressive correlation structure
\[
\Sigma_{ij}=\rho^{|i-j|},
\qquad
\rho\in\{0,0.1\}.
\]
After generation, covariates were standardized within each simulated data set.
In selected exploratory scenarios, a subset of covariates was dichotomized to
obtain mixed continuous and binary covariate structures.

\outsourceflag{coefficient-pattern details and table}
\begin{outsourceblock}
The true regression parameter \(\beta^\circ\) followed one of four predefined
coefficient patterns representing different signal strengths and sparsity
levels, following the simulation design of Kammer et al.\
\cite{Kammer2022LassoSelective}. Specifically, we considered an all-ones
pattern, a high-contrast pattern, a realistic sparse pattern, and a very sparse
pattern (Table~\ref{tab:beta_patterns}). Coefficients not explicitly specified
in the corresponding pattern were set to zero.

\begin{table}[htbp]
  \centering
  \caption{Coefficient patterns used for the data-generating vector
  \(\beta^\circ\).}
  \label{tab:beta_patterns}
  \begin{tabularx}{0.9\textwidth}{@{} l l X @{}}
    \toprule
    \textbf{Pattern} & \textbf{Coefficient vector} & \textbf{Description} \\
    \midrule
    allones
      & \((1,1,\ldots,1)\)
      & Equal effects on all covariates. \\
    highcontrast
      & \((0.3,1.0,0.3,1.0,0,\ldots,0)\)
      & Alternating small and large effects. \\
    realistic
      & \((0.8,0.7,0.5,0.8,0,\ldots,0)\)
      & Moderate effects on four covariates. \\
    sparse
      & \((1,1,0,\ldots,0)\)
      & Two active coefficients. \\
    \bottomrule
  \end{tabularx}
\end{table}
\end{outsourceblock}

\outsourceflag{detailed event-time generation formulas}
\begin{outsourceblock}
Event times were generated from a Cox proportional hazards model with linear
predictor \(X^\top\beta^\circ\). Conditional on \(X\), survival times satisfy
\[
\Pr(T>t\mid X)
=
\exp\left\{-H_0(t)\exp(X^\top\beta^\circ)\right\}.
\]
Thus, using inverse transform sampling,
\[
T
=
H_0^{-1}\left(
\frac{-\log U}{\exp(X^\top\beta^\circ)}
\right),
\qquad
U\sim \mathrm{Unif}(0,1).
\]
Two baseline cumulative hazard functions were considered. For the exponential
baseline, \(H_0(t)=t\), giving
\[
T
=
\frac{-\log U}{\exp(X^\top\beta^\circ)}.
\]
For the Weibull baseline, \(H_0(t)=s t^k\), with \(s=1\) and \(k=2\), giving
\[
T
=
\left\{
\frac{-\log U}{s\exp(X^\top\beta^\circ)}
\right\}^{1/k}.
\]
\end{outsourceblock}

\outsourceflag{censoring, jitter, and tuning implementation details}
\begin{outsourceblock}
\begingroup
Right censoring was imposed through a non-informative administrative censoring
mechanism following the sampling idea of \citet{ramos2024sampling}. For each
Monte Carlo repetition, event times \(\{T_i\}_{i=1}^n\) were generated first.
For a target censoring proportion
\[
t_C\in\{0,0.10,0.30\},
\]
we then defined a censoring cutoff \(c^\star\) as the
\((1-t_C)\)-empirical quantile of the generated event times. For \(t_C=0\), no
administrative censoring was applied. Observed times and event indicators were
defined as
\[
Y_i=\min(T_i,c^\star),
\qquad
\delta_i=\mathbbm{1}(T_i\le c^\star).
\]
This construction yields censoring proportions close to the target level by
design.

To avoid numerical problems due to ties, small deterministic jitters were added
only to duplicated event and, where applicable, censoring times. For each group
of identical times, offsets of order \(10^{-8}\) were added to event times and
offsets of order \(5\cdot 10^{-9}\) were added to censoring times, each scaled
by the range of the corresponding time variable. These perturbations were used
only for numerical stability and preserve the ordering of observations up to
negligible changes.

For each simulated data set, Cox--Lasso tuning was performed using two
implementable penalty-selection rules: an AIC-based rule and the
\(\lambda_{\min}\) rule. The AIC-based rule selected the penalty by minimizing
an information criterion along the Cox--Lasso solution path. The
\(\lambda_{\min}\) rule selected the value of \(\lambda\) that minimized the
mean cross-validated Cox partial-likelihood deviance. Thus, among the
cross-validated candidate penalties, \(\lambda_{\min}\) corresponds to the
lowest estimated prediction error rather than to the more regularized
one-standard-error choice. These rules were included to assess whether
post-selection interval performance is sensitive to the amount of penalization
induced by the tuning choice.
% No penalty level was selected using knowledge of the true active set, the
% simulation target, or empirical coverage.
\endgroup
\end{outsourceblock}

\iffalse
\rev{For each simulated data set, Cox--Lasso tuning was performed using two
implementable penalty-selection rules: an AIC-based rule and the
\(\lambda_{\min}\) rule. The AIC-based rule selects the penalty by minimizing
an information criterion along the Cox--Lasso path, whereas the
\(\lambda_{\min}\) rule selects the penalty value corresponding to the minimum
of the implemented tuning criterion. These two rules were included to assess
whether post-selection interval performance is sensitive to the amount of
penalization induced by the tuning choice. No penalty level was chosen using
knowledge of the true active set or empirical coverage.}
\fi 

The primary design factors are summarized in Table~\ref{tab:design_factors}.

\begin{table}[htbp]
\centering
\caption{Primary design factors in the simulation study.}
\label{tab:design_factors}
\small
\begin{tabular}{ll}
\toprule
Factor & Levels \\
\midrule
Sample size \(n\) & \(30,40,60,75,80,100,125,150,175,200,250\) \\
Number of covariates \(p\) & \(10,20\) \\
Correlation \(\rho\) & \(0,0.1\) \\
Baseline hazard & exponential, Weibull \((k=2,s=1)\) \\
Coefficient pattern & allones, highcontrast, realistic, sparse \\
Target censoring proportion \(t_C\) & \(0,0.10,0.30\) \\
\rev{Penalty tuning rule} & \rev{AIC, \(\lambda_{\min}\)} \\
Nominal confidence level & \(90\%\) \\
Monte Carlo repetitions & \(900\) \\
Bootstrap repetitions & \(B=500\) primary; \(B\in\{200,900\}\) sensitivity analyses \\
\bottomrule
\end{tabular}
\end{table}
\subsection{Estimands and performance measures}
\label{sec:estimand}

\begin{newrevblock}
We use the term \emph{estimand} for the population quantity that a confidence
interval is intended to cover and distinguish it from the estimator calculated
from a data set and from the empirical criterion used to evaluate the interval
in the simulation study. In the present simulation, the data-generating Cox
coefficient \(\beta_j^\circ\) serves as the common reference estimand for
variable \(j\), but the compared methods address this coefficient within
different inferential frameworks.

Let \(\widehat M\) denote the variable set selected by the Cox--Lasso in an
original simulated data set. For the model-based bootstrap, the reported
estimator is the unpenalized Cox refit coefficient
\(\widetilde\beta_{\widehat M,j}\), and an interval is constructed only when
\(j\in\widehat M\). The raw penalized Cox--Lasso coefficient is not used for
inference. Under the oracle regime considered in
Section~\ref{sec:asymptotics}, the event \(\widehat M=M^\circ\) has probability
tending to one, and the unpenalized Cox refit estimator is consistent for
\(\beta_j^\circ\) for active variables.

As a comparison method, we consider a debiased Cox--Lasso estimator. This
method adds a bias-correction term to the penalized estimator and is designed
for componentwise inference on a fixed coordinate \(\beta_j^\circ\) of the
full high-dimensional parameter, irrespective of whether that variable was
selected by the Cox--Lasso. Under the oracle-type selection-stability assumptions used in Section~\ref{sec:asymptotics}, \(\mathbb P(\widehat M=M^\circ)\to1\). Hence, for active variables, the unpenalized Cox refit asymptotically concerns the
same data-generating coefficient \(\beta_j^\circ\) as the debiased estimator.
This common asymptotic target makes the debiased estimator an informative
benchmark for the proposed procedure. Nevertheless, the two methods address different finite-sample inferential problems. The model-based bootstrap concerns coefficients reported after variable selection and is evaluated conditional on the variable being selected, whereas the debiased method concerns a prespecified coordinate of the full parameter vector and does not require prior selection.
\end{newrevblock}

Because \(\beta_j^\circ\) is known in the simulation, it is used as the
reference value whenever variable \(j\) is selected in the original simulated
data set. The primary performance measure is the selection-conditional coverage
probability
\[
\mathrm{SCov}^{\mathrm{sim}}_j
=
\mathbb{P}\!\left(
\beta_j^\circ\in\mathrm{CI}_j
\mid
j\in\widehat M
\right).
\]
Thus, coverage is evaluated only in Monte Carlo repetitions in which \(j\) is
selected. This criterion describes the empirical operating characteristics of
intervals reported after variable selection. Comparisons with debiased
intervals should therefore be interpreted as comparisons under a common
data-generating reference value, but not under a common selection or
conditioning framework.

Secondary performance measures are mean confidence-interval width,
coefficient-averaged active and inactive coverage, and runtime; their empirical
estimators are given in Supporting Information~B.2. Coverage is treated as a
simulation performance measure in the sense of \citet{morris2019using} and is
evaluated conditionally on original-sample selection, as in comparative studies
of selective confidence intervals
\citep{Kammer2022LassoSelective}.
\subsection{Methods compared}

We compare conventional post-selection Cox inference based on an unpenalized
\texttt{coxph} refit on the selected support and the corresponding
normal-approximation Wald intervals; debiased Cox--Lasso inference as a
high-dimensional reference method; and \plrev{the proposed model-based
bootstrap with percentile, Wald-type, and studentized intervals.} The oracle
Cox estimator fitted on the true active set is included as an unattainable
benchmark representing correct model specification without selection error.

Cox--Lasso tuning parameters were selected in the original sample using either
an AIC-based rule or the \(\lambda_{\min}\) rule. The comparison between these
rules assesses sensitivity to an implementable tuning choice; no penalty was
selected using the true active set or empirical coverage.

\plrev{The model-based bootstrap was implemented as described in
Section~\ref{sec:mbboot} and Algorithm~\ref{alg:coxplugin}. The support selected
in the original sample was reused in all bootstrap samples, and event times
were generated from the fitted post-Lasso Cox model. Because censoring was
covariate-independent in the simulation design, bootstrap censoring times were
drawn from the Kaplan--Meier estimator fitted to
\((Y_i,1-\delta_i)\), treating censoring as the event of interest.}

\subsection{Results}
\label{sec:results}

\plrev{We report the finite-sample performance of naive post-selection Wald
inference, debiased Lasso inference, and the model-based bootstrap.}

The procedures are compared with respect to selective coverage, interval
length, and runtime, focusing on systematic patterns across sample size and
censoring.
\rev{All reported coverage values are evaluated conditional on selection in the
original sample. This makes the bootstrap variants directly comparable to each
other, but it means that the debiased estimator is evaluated under a
post-selection conditioning event that is not its primary theoretical target.}
%Methods are compared with respect to selective coverage,
%average interval length, and computational cost.
%Particular attention is given to the behavior across
%increasing sample size and varying censoring proportions,
%as these dimensions are most relevant for assessing
%finite-sample improvements beyond first-order asymptotic.
%\subsection{Selective coverage}

%Figure~\ref{fig:coverage} displays empirical selective coverage for a
%representative scenario with $p=20$.
%Naive Wald inference exhibits substantial u%ndercoverage in small
%samples, reflecting failure to account for the selection step.
%Debiased Lasso improves coverage relative to Wald inference but
%continues to show deviations from the nominal 90\% level when
%sample sizes are small or censoring is moderate.

%The proposed bootstrap procedure attains coverage closer to the
%nominal level across most scenarios.
%Coverage deviations decrease with increasing sample size for all
%methods, consistent with asymptotic arguments.

Additional results can be found in Supporting Information~B.3.
\begin{newrevblock}
Results are displayed separately for two implementable tuning strategies.
AIC-based tuning selects the value of \(\lambda\) that minimizes the AIC along
the Cox--Lasso solution path. Under cross-validation-based tuning,
\(\lambda_{\min}\) denotes the value of \(\lambda\) that minimizes the
\(K\)-fold cross-validation error. Both strategies use the same
\(\ell_1\)-penalty, but the selected value of \(\lambda\) determines the
effective strength of regularization and thereby affects the size and stability
of the selected variable set.

Neither tuning strategy uses knowledge of the true active set or of the value
of \(\lambda\) that would optimize empirical coverage. The comparison
therefore represents a sensitivity analysis based on practically implementable
tuning strategies rather than a comparison with an oracle choice of
\(\lambda\). The inferential estimand remains the data-generating coefficient
\(\beta_j^\circ\), while the tuning strategy determines which variables are
selected and hence which coefficients enter the selection-conditional coverage
evaluation.
\end{newrevblock}
\begin{newrevblock}
The figures include percentile, Wald-type, and studentized bootstrap intervals.
Across the displayed scenarios, the Wald-type bootstrap interval was the most
consistently problematic bootstrap variant, particularly in small samples and
for weaker active effects. In several scenarios, this pattern was more
pronounced under cross-validation-based tuning. This difference should not be
interpreted as arising from a different type of penalty; rather, the two tuning
strategies selected different values of \(\lambda\) and consequently different
variable sets.
The comparatively poor performance of the Wald-type interval is consistent
with its reliance on a symmetric normal approximation and on the empirical
bootstrap standard deviation as an estimate of sampling variability.
Percentile and studentized intervals were generally more stable and showed
similar empirical performance in several settings. This suggests that the main
finite-sample improvement arose from using quantiles of the bootstrap
distribution rather than normal critical values.
\end{newrevblock}
\rev{Relative to the debiased estimator, the model-based bootstrap showed a
mixed finite-sample pattern. It did not uniformly dominate the debiased approach
across scenarios. Its empirical coverage was most sensitive to sample size and
penalty tuning for weaker coefficients. These findings are important because
they identify the finite-sample regimes in which the proved bootstrap implementation
is useful, and also the regimes in which first-order theory does not eliminate
all post-selection distortion.}

\outsourceflag{coefficient-averaged active-coverage summary (not coefficient-specific)}
\begin{outsourceblock}
\outsource{To complement the coefficient-specific plots, Table~\ref{tab:coverage_active_by_n_main}
summarizes coefficient-averaged empirical selective coverage over truly active
coefficients for a representative setting (Weibull baseline, realistic
coefficient pattern, $t_C = 0.1$, $\rho=0$, $p=10$, $B=500$) across several
sample sizes up to $n=250$. For each active coefficient, coverage was computed
conditional on selection in the original sample and then averaged over active
coefficients. Similar coefficient-averaged summaries have been used in
simulation studies of post-selection inference, for example when reporting
coverage for confidence intervals of selected coefficients averaged across
coefficients with comparable effect size \citep{garciarasines2023splitting}.
However, because coverage is formally a coefficient-specific property and
selection effects can depend on signal strength and selection frequency
\citep{benjamini2005fdr,lee2016exact}, the table is intended only as a
descriptive summary of the sample-size trend and not as a replacement for the
coefficient-specific coverage assessment. In this averaged active-coefficient
summary, model-based bootstrap percentile and studentized coverage increase with sample
size but remain below nominal for the smaller samples. This pattern is reported
as a finite-sample limitation of this implementation. Consistent with the broader simulation results,
Wald-type intervals are not emphasized in the main table, whereas percentile and
studentized intervals are shown as the more relevant bootstrap variants.
Figure~\ref{fig:ci_width} reports the corresponding mean interval widths for
active coefficients across $n$.}

\begin{table}[htbp]
\centering
\scriptsize
\begin{tabular}{lrrrrr}
\toprule
Method & $n=50$ & $n=100$ & $n=150$ & $n=200$ & $n=250$ \\
\midrule
Debiased & 0.733 & 0.786 & 0.851 & 0.829 & 0.840 \\
Param Perc-Int & 0.615 & 0.666 & 0.716 & 0.752 & 0.818 \\
Param Std-Int & 0.614 & 0.666 & 0.718 & 0.756 & 0.814 \\
\bottomrule
\end{tabular}
\caption{\outsource{Coefficient-averaged empirical selective coverage for truly active
coefficients across sample sizes in the highlighted Weibull/realistic scenario
($p=10$, $t_C=0.1$, $\rho=0$, $B=500$). Coverage was first computed
separately for each active coefficient conditional on original-sample selection
and was then averaged over active coefficients. }}
\label{tab:coverage_active_by_n_main}
\end{table}
\end{outsourceblock}

%\plrev{Coefficient-specific diagnostics for the highlighted scenario are summarized through the coverage displays below. The target-1-only version of the manuscript omits the earlier diagnostic figure that included additional bootstrap variants outside the scope of this manuscript.}

%\begin{figure}[h]
%    \centering
%    \includegraphics[width=.7\linewidth]{manuscript_x3_lambda_n100.pdf}
%    \caption{Empirical coverage probabilities of confidence intervals for coefficient $\beta_3$ as a function of the sample size $n$ for $p=10$ and $p=20$. Panels correspond to different choices of the regularization parameter selection rule (AIC, Min, Fixed). Results are shown for the oracle estimator, the debiased estimator, and three bootstrap-based intervals (Percentile, Studentized, and Wald). The dashed horizontal line indicates the nominal coverage level. Only sample sizes $n \le 100$ are displayed. \rev{The fixed-penalty panel is included as a sensitivity display and is not interpreted as an oracle choice of the true coverage-optimal penalty.}}
%    \label{fig:coverage_x3_lambda_n100}
%\end{figure}
\begin{figure}[h]
    \centering
    \includegraphics[width=.7\linewidth]{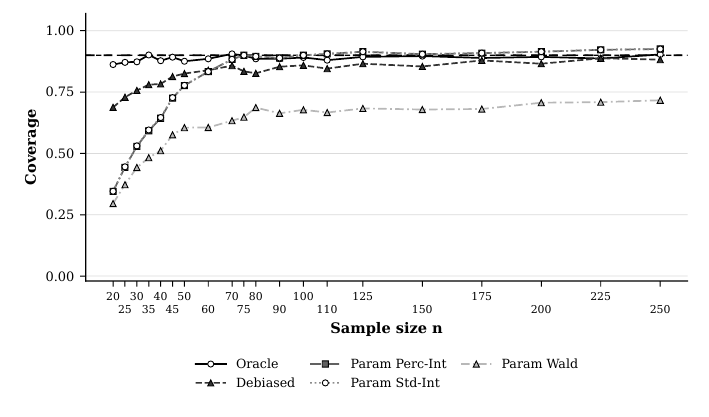}
    \caption{Coefficient-specific empirical selective coverage for coefficient
    \(\beta_2\) in the Weibull/realistic scenario with the \(\lambda_{\min}\)
    tuning rule, target censoring 0.1, \(\rho=0\), 900 Monte Carlo repetitions,
    and \(B=200\) bootstrap repetitions. The display is shown up to \(n=250\).
    The dashed horizontal line denotes the nominal 90\% confidence level.
    Results are shown for the oracle estimator, the debiased estimator, and
    percentile, studentized, and Wald-type parametric bootstrap intervals.}
    \label{fig:coverage_beta2_old}
\end{figure}

\rev{Figure~\ref{fig:coverage_beta2_old} provides a coefficient-specific
coverage diagnostic for the highlighted Weibull/realistic scenario. For very
small samples, percentile and studentized bootstrap intervals can still
under-cover, reflecting the difficulty of conditional post-selection statements
when selection is unstable. As sample size increases, these intervals move
towards the nominal level and are closer to nominal than the bootstrap-Wald
interval over a substantial range of moderate sample sizes.}

\outsourceflag{full numerical censoring-sensitivity table}
\begin{outsourceblock}
\outsource{Table~\ref{tab:beta3_censoring_coverage_values_main} reports the
corresponding numerical censoring-sensitivity results for \(\beta_3\) under the
\(\lambda_{\min}\) rule. Replacing the plot by the table makes the finite-sample
pattern explicit. For \(p=10\), model-based bootstrap percentile and studentized
coverage increase with sample size and approach the nominal level by \(n=100\),
although coverage is lower at higher censoring levels in the small samples. For
\(p=20\), undercoverage is more pronounced, particularly at \(t_C=0.3\), where
coverage remains unstable even at the largest displayed sample sizes. The
Wald-type bootstrap interval is consistently the lowest-coverage bootstrap
variant, whereas the debiased estimator is less sensitive to censoring but still
mostly remains below the nominal level. Thus, the censoring results reinforce
that the model-based bootstrap improves with sample size but does not remove
all finite-sample effects, especially when censoring and model dimension make
selection less stable.}

\begin{table}[h]
\centering
\scriptsize
\resizebox{\textwidth}{!}{%
\begin{tabular}{lllrrrrrr}
\toprule
\(p\) & \(t_C\) & Method & \(n=30\) & \(n=40\) & \(n=60\) & \(n=75\) & \(n=80\) & \(n=100\) \\
\midrule
10 & 0.0 & Oracle & 0.862 & 0.899 & 0.887 & 0.888 & 0.876 & 0.903 \\
 &  & Debiased & 0.763 & 0.825 & 0.830 & 0.846 & 0.843 & 0.875 \\
 &  & Param Perc-Int & 0.524 & 0.680 & 0.851 & 0.864 & 0.916 & 0.922 \\
 &  & Param Std-Int & 0.524 & 0.681 & 0.851 & 0.864 & 0.916 & 0.922 \\
 &  & Param Wald & 0.424 & 0.572 & 0.703 & 0.705 & 0.737 & 0.756 \\
\addlinespace
10 & 0.1 & Oracle & 0.864 & 0.887 & 0.891 & 0.878 & 0.877 & 0.900 \\
 &  & Debiased & 0.783 & 0.819 & 0.835 & 0.833 & 0.839 & 0.866 \\
 &  & Param Perc-Int & 0.509 & 0.619 & 0.756 & 0.776 & 0.846 & 0.908 \\
 &  & Param Std-Int & 0.509 & 0.619 & 0.756 & 0.776 & 0.846 & 0.908 \\
 &  & Param Wald & 0.401 & 0.514 & 0.644 & 0.642 & 0.695 & 0.749 \\
\addlinespace
10 & 0.3 & Oracle & 0.864 & 0.883 & 0.889 & 0.885 & 0.898 & 0.907 \\
 &  & Debiased & 0.767 & 0.795 & 0.821 & 0.842 & 0.870 & 0.861 \\
 &  & Param Perc-Int & 0.448 & 0.556 & 0.694 & 0.771 & 0.801 & 0.868 \\
 &  & Param Std-Int & 0.448 & 0.556 & 0.694 & 0.771 & 0.801 & 0.868 \\
 &  & Param Wald & 0.348 & 0.458 & 0.584 & 0.661 & 0.674 & 0.730 \\
\addlinespace
20 & 0.0 & Oracle & 0.882 & 0.872 & 0.882 & 0.898 & 0.881 & 0.903 \\
 &  & Debiased & 0.739 & 0.743 & 0.807 & 0.833 & 0.827 & 0.851 \\
 &  & Param Perc-Int & 0.370 & 0.501 & 0.690 & 0.778 & 0.815 & 0.875 \\
 &  & Param Std-Int & 0.371 & 0.501 & 0.690 & 0.778 & 0.815 & 0.875 \\
 &  & Param Wald & 0.245 & 0.357 & 0.496 & 0.578 & 0.594 & 0.632 \\
\addlinespace
20 & 0.1 & Oracle & 0.860 & 0.864 & 0.877 & 0.881 & 0.883 & 0.892 \\
 &  & Debiased & 0.754 & 0.751 & 0.818 & 0.830 & 0.825 & 0.844 \\
 &  & Param Perc-Int & 0.357 & 0.479 & 0.664 & 0.753 & 0.785 & 0.843 \\
 &  & Param Std-Int & 0.357 & 0.479 & 0.665 & 0.753 & 0.785 & 0.843 \\
 &  & Param Wald & 0.223 & 0.338 & 0.497 & 0.554 & 0.582 & 0.616 \\
\addlinespace
20 & 0.3 & Oracle & 0.860 & 0.865 & 0.890 & 0.888 & 0.910 & 0.899 \\
 &  & Debiased & 0.813 & 0.775 & 0.811 & 0.834 & 0.844 & 0.843 \\
 &  & Param Perc-Int & 0.287 & 0.403 & 0.589 & 0.812 & 0.715 & 0.706 \\
 &  & Param Std-Int & 0.287 & 0.403 & 0.589 & 0.812 & 0.715 & 0.706 \\
 &  & Param Wald & 0.189 & 0.284 & 0.448 & 0.612 & 0.543 & 0.492 \\
\bottomrule
\end{tabular}%
}
\caption{\outsource{Numerical empirical selective coverage values for the
censoring-sensitivity analysis of coefficient \(\beta_3\). The table reports the
Weibull/realistic scenario with the \(\lambda_{\min}\) tuning rule, \(\rho=0\),
target censoring levels \(t_C\in\{0,0.1,0.3\}\), \(p\in\{10,20\}\), and
model-based bootstrap intervals. %Where duplicate runs were available for the same configuration, coverage was aggregated using the corresponding coverage denominators.
}}
\label{tab:beta3_censoring_coverage_values_main}
\end{table}
\end{outsourceblock}

\outsourceflag{interval-width display}
\begin{outsourceblock}
\outsource{To summarize interval width more compactly, Figure~\ref{fig:ci_width}
shows box plots of the coefficient-wise mean confidence-interval widths for the
same highlighted scenario. This representation emphasizes the width
distribution across active coefficients and sample sizes rather than only the
sample-size trend. The model-based bootstrap percentile and studentized intervals are
wider than the Wald-type bootstrap intervals, which is consistent with their
improved coverage. \outsource{The width display refers to the implementation described in Algorithm~\ref{alg:coxplugin}.}}

\begin{figure}[h]
    \centering
    \includegraphics[width=.76\linewidth]{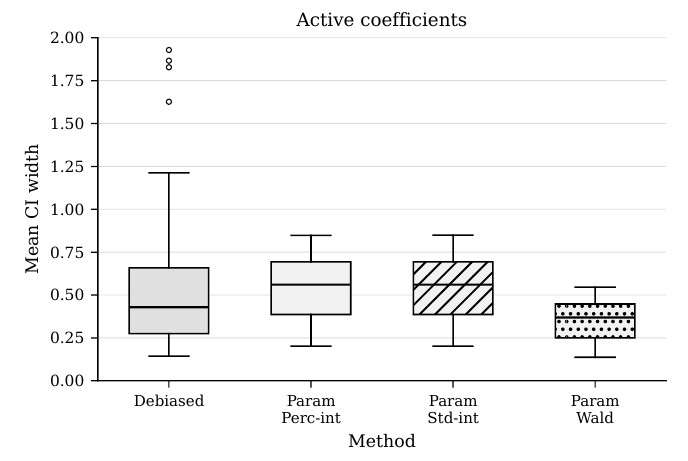}
    \caption{\outsource{Distribution of coefficient-wise mean confidence-interval
    widths for active coefficients in the highlighted Weibull/realistic
    scenario ($p=10$, $t_C=0.1$, $\rho=0$, $B=500$).}}
    \label{fig:ci_width}
\end{figure}
\end{outsourceblock}

%\subsection{Selective type~I error}

%Selective type~I error is inflated for naive Wald inference in
%settings with moderate correlation or censoring.
%Debiased Lasso improves control but exhibits mild inflation in
%smaller samples.
%The bootstrap procedure maintains rejection rates closer to the
%nominal level across most scenarios.

%\subsection{Computational performance}

\rev{Runtime per simulation replication was benchmarked for the naive Wald
procedure and the model-based bootstrap implementation. The benchmark used
\(B=1000\) bootstrap replications per simulation run. In addition to the primary
dimensions \(p=10\) and \(p=20\), larger values of \(p\) were included to
illustrate computational scaling. The model-based bootstrap was substantially
more expensive than the naive procedure but remained practically feasible across
the reported dimensions. Figure~\ref{fig:runtime_old} summarizes the runtime
benchmark; exact numerical runtime summaries are provided in the Supporting
Information.}

\begin{figure}[h]
    \centering
    \includegraphics[width=.76\linewidth]{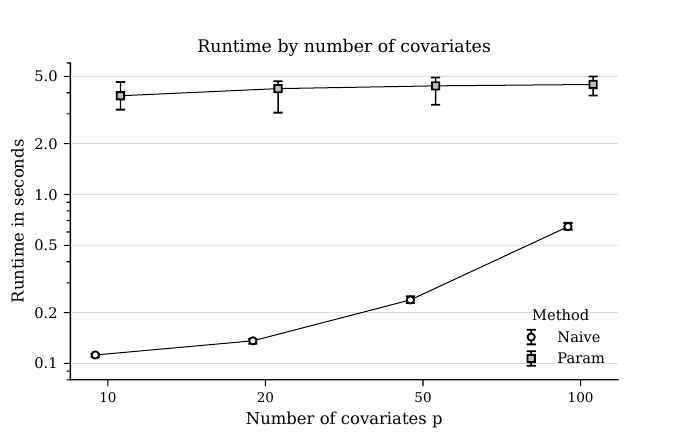}
    \caption{Runtime benchmark for the conventional post-selection Wald
    procedure and the parametric bootstrap. Points show median runtimes and
    intervals summarize the interquartile range across repeated benchmark runs;
    the vertical axis is shown on a log scale to display both procedures in the
    same panel.}
    \label{fig:runtime_old}
\end{figure}

\iffalse
\begin{table}[htbp]
\centering
\begin{tabular}{lrrrr}
\toprule
p & Method & Median seconds & Q25 & Q75 \\
\midrule
10 & Naive & 0.112 & 0.110 & 0.115 \\
10 & Param & 3.838 & 3.179 & 4.633 \\
20 & Naive & 0.135 & 0.131 & 0.140 \\
20 & Param & 4.231 & 3.052 & 4.682 \\
50 & Naive & 0.238 & 0.228 & 0.250 \\
50 & Param & 4.393 & 3.400 & 4.918 \\
100 & Naive & 0.644 & 0.619 & 0.679 \\
100 & Param & 4.476 & 3.854 & 4.989 \\
\bottomrule
\end{tabular}
\caption{\rev{Absolute runtime summary for the parametric bootstrap benchmark. Values are medians with interquartile ranges over repeated runs.}}
\label{tab:runtime_absolute_summary}
\end{table}
\fi
%\subsection{Summary}

\rev{Across simulation scenarios, the proposed bootstrap procedure
improves selective coverage relative to the most naive Wald-type constructions,
but the finite-sample pattern is not uniform across coefficients, tuning rules,
and censoring levels. In the focused Weibull/realistic summary shown in
Table~\ref{tab:coverage_active_by_n_main}, the model-based bootstrap variants
show increasing active-coefficient coverage with sample size.\plrev{The remaining undercoverage observed in the smallest samples is
reported as a finite-sample limitation of the proposed bootstrap procedure.
Performance differences generally decreased with increasing sample size,
consistent with the first-order asymptotic result established in
Section~\ref{sec:asymptotics}.}}

\section{Real data example}
\label{sec:example}

\begingroup

This section applies the proposed bootstrap procedure to a real survival
dataset. The purpose of the example is methodological: we use a realistic
selected Cox model to assess the behavior of the model-based
bootstrap after variable selection via the Cox--Lasso. In particular, the
example illustrates that the procedure can be implemented in a real
post-selection survival analysis and compares the resulting intervals with a
standard post-selection Wald approximation. \plrev{The inferential target is
the selected-support Cox coefficient: the support is selected once in the
original data and then kept fixed when constructing the bootstrap intervals.}
This choice is motivated by the post-selection inference literature, which
shows that conventional inferential statements after data-driven model
selection can be anti-conservative when selection uncertainty is ignored
\cite{Berk2013PostSelection,lee2016exact,LeebPotscher2005,LeebPotscher2006,taylor2018post}.

\plrev{The real data example is not intended to establish definitive clinical
conclusions about breast cancer prognosis. Instead, it illustrates how the
resulting post-selection intervals behave in a realistic selected Cox model
containing a mixture of strong and weaker prognostic variables.}

\subsection{Data description}

The analysis is based on data from the Surveillance, Epidemiology,
and End Results (SEER) Program of the National Cancer Institute
\cite{seer_program}. We considered female patients
diagnosed with primary invasive breast cancer between 2010 and 2015.
To obtain a computationally manageable yet clinically homogeneous
analysis cohort, we restricted the sample to stage II--III disease and
drew a random subsample of 8000 individuals after preprocessing.

The primary endpoint was overall survival, defined as time from
diagnosis to death from any cause. Individuals alive at last follow-up
were treated as right-censored. The final analytic sample comprised
$n=8000$ patients, of whom 2489 experienced the event. The median
follow-up time was 102 months, and 68.9\% of observations were
censored. 
\plrev{We note that the censoring proportion in this registry-based example is
higher than in the main simulation scenarios. The analysis should therefore be
read as an illustration of the implementation and qualitative interval behavior
in a realistic post-selection Cox analysis with substantial right censoring,
rather than as an additional censoring-sensitivity study.}

\begin{table}[t]
\centering
\caption{Baseline characteristics of the SEER breast cancer analysis sample
(\(n=8000\)).}
\label{tab:seer_baseline}
\begin{tabular}{l r}
\toprule
\textbf{Characteristic} & \textbf{Summary} \\
\midrule
\multicolumn{2}{l}{\emph{Demographics}} \\
Age at diagnosis, years, mean (SD) & 58.8 (14.2) \\
White race, \% & 76.5 \\
Married, \% & 53.5 \\[0.3em]

\multicolumn{2}{l}{\emph{Tumor characteristics}} \\
Stage II, \% & 74.4 \\
Stage III, \% & 25.6 \\
High grade (III/IV), \% & 41.0 \\
Hormone receptor positive, \% & 77.8 \\[0.3em]

\multicolumn{2}{l}{\emph{Treatment}} \\
Surgery performed, \% & 94.3 \\
Radiation therapy, \% & 51.8 \\
Chemotherapy, \% & 63.6 \\[0.3em]

\midrule
Events, \(n\) & 2489 \\
Censored observations, \% & 68.9 \\
Median follow-up, months & 102 \\
\bottomrule
\multicolumn{2}{l}{\footnotesize Values are mean (SD) for age, percentages
unless otherwise indicated,} \\
\multicolumn{2}{l}{\footnotesize and counts where explicitly denoted by \(n\).}
\end{tabular}
\end{table}

The final design matrix contained $p=74$ covariates after preprocessing
and dummy coding. The covariates covered demographic characteristics,
tumor stage and grade, receptor status, tumor size, nodal involvement,
treatment indicators, rural--urban status, and income categories.
Continuous variables were standardized prior to model fitting, and
categorical variables were encoded using dummy variables.

\iffalse 
\subsection{Research question}

The objective is to assess the behavior of the proposed
model-based bootstrap after variable selection via the Cox--Lasso. The real data
example is not intended to establish definitive clinical conclusions
about breast cancer prognosis. Instead, it evaluates how the proposed
post-selection bootstrap intervals behave in a realistic survival
analysis in which the selected model contains a mixture of strong and
weaker prognostic variables.

\fi

\subsection{Results}

\begingroup
Cross-validated Cox--Lasso selection resulted in an active model containing
20 covariates. The selected variables included age, receptor-status indicators,
nodal involvement, tumor-stage indicators, tumor grade, race indicators, and
treatment variables. Figure~\ref{fig:example} summarizes 90\% intervals for the
selected coefficients. The display compares naive post-selection Wald inference
based on the post-Lasso Cox refit, the proposed model-based
bootstrap intervals, and a debiased Cox--Lasso interval as a
reference method. All coefficients are shown on the normalized design-matrix
scale used for model fitting.

The model-based bootstrap is the procedure aligned with the
theoretical target studied in this manuscript. It quantifies uncertainty for the Cox
coefficients reported after the original Cox--Lasso selection step, with the
selected support kept fixed throughout the bootstrap. The debiased interval is
included only as a reference and should not be interpreted as
target-equivalent to the bootstrap intervals. In particular, debiased Cox--Lasso
inference is motivated by componentwise inference for fixed
coordinates of the regression parameter, whereas the bootstrap intervals are
constructed for the selected-support Cox refit. Consequently, visible
differences between the debiased reference and the bootstrap intervals mainly
reflect differences in inferential target rather than a direct failure of one
interval construction.

Across the selected covariates, the model-based bootstrap
intervals led to similar qualitative conclusions for the strongest effects.
Chemotherapy, positive lymph nodes, PR negativity, absence of surgery, and
AJCC~T4 remained clearly separated from zero. Weaker selected variables, such as
race indicators, tumor size, and some AJCC nodal categories, showed larger
sensitivity to the chosen uncertainty construction. This is the part of the real
data example where the comparison is most informative: it illustrates how
standard post-selection Wald intervals, model-based bootstrap
intervals, and a debiased reference can differ for weaker
selected effects.

The interval displays were regenerated for \(B=300\), \(B=500\), and \(B=900\)
bootstrap replications. The qualitative pattern was stable across these
bootstrap budgets; the \(B=900\) result is shown in the main text, and the
corresponding \(B=300\) and \(B=500\) displays are reported in the
Supporting Information.
\endgroup

\begin{figure}[t]
    \centering
    \includegraphics[width=.98\linewidth]
    {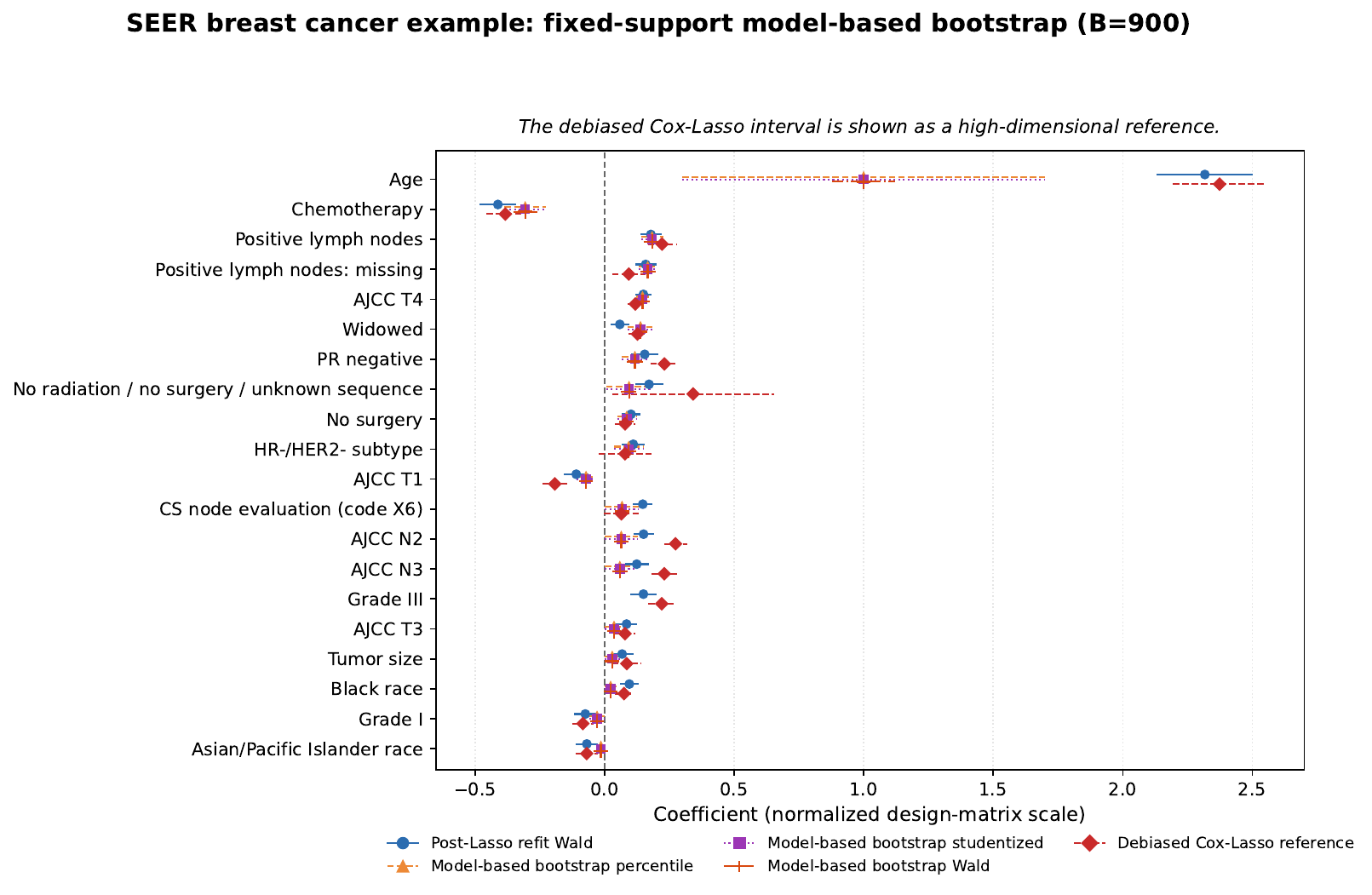}
    \caption{\newrev{Coefficient estimates and 90\% confidence intervals for
    variables selected by the Cox--Lasso in the SEER breast cancer example.
    The five displayed interval constructions are: (i) the
    normal-approximation Wald interval based on the unpenalized Cox refit,
    (ii) the model-based bootstrap percentile interval,
    (iii) the model-based bootstrap Wald-type interval,
    (iv) the model-based bootstrap studentized interval, and
    (v) the debiased Cox--Lasso reference interval. The three model-based
    bootstrap intervals correspond to the procedure in which the selected
    variable set is held fixed across bootstrap samples. Under the oracle-type
    assumptions, the bootstrap refit and the debiased estimator concern the
    same data-generating coefficients for active variables, but they address
    different finite-sample inferential problems. Coefficients are shown on the
    normalized design-matrix scale. The bootstrap intervals are based on
    \(B=900\) replications.}}
    \label{fig:example}
\end{figure}

\begingroup
Overall, the example demonstrates that the model-based bootstrap
can be applied to a realistic post-selection Cox model and yields interpretable
interval estimates for coefficients reported after Cox--Lasso selection.
Because the SEER analysis is an illustrative registry-based data example with
substantial right censoring, the numerical results should not be read as a
separate validation of coverage in this censoring regime. Rather, the example
shows how the interval constructions behave in one realistic analysis
containing both strong and weaker selected prognostic variables. The main
empirical message is the difference between standard post-selection Wald
inference, model-based bootstrap uncertainty quantification, and a
debiased reference.
\endgroup

\iffalse
Earlier no-debiased SEER block from the Biometrical-Journal draft was removed
from compilation because the current SEER analysis uses the newer figure with
the debiased reference. The restored old draft figures are still kept in the ZIP.
\fi

\endgroup

\section{Discussion}
\label{sec:discussion}

This work studies model-based bootstrap procedures for post-selection
inference in Cox regression. We establish first-order bootstrap validity in
the fixed-dimensional setting and extend the result to the
$\ell_1$-penalized case through an oracle reduction argument. The theoretical
analysis yields conditional weak convergence of the centered and scaled
bootstrap estimator and thereby implies asymptotically correct coverage for
standard bootstrap confidence intervals.
\newrev{The theoretical guarantee applies to the procedure in which the
Cox--Lasso selects a variable set once in the original data and the bootstrap
approximates the sampling distribution of the unpenalized Cox refit while
holding that selected variable set fixed.}
Studentized intervals are covered by additionally requiring consistency of the
usual Cox information-based standard-error estimator in the \newrev{unpenalized Cox refit}.

\begin{newrevblock}
The simulation study complements the asymptotic theory by examining
finite-sample performance across sample sizes, correlation structures,
censoring levels, and two implementable tuning strategies. In the first
strategy, \(\lambda\) is selected by minimizing the AIC along the Cox--Lasso
solution path. In the second, \(\lambda_{\min}\) denotes the value of
\(\lambda\) that minimizes the \(K\)-fold cross-validation error. Both
strategies use the same \(\ell_1\)-penalty; they differ only in how the
regularization parameter \(\lambda\) is selected.

The proposed bootstrap procedure improves coverage relative to naive
post-selection Wald intervals in many settings, although the magnitude of the
improvement depends on the coefficient, sample size, censoring level, and
tuning strategy. Coverage differences between the methods generally became
smaller with increasing sample size, consistent with the first-order
asymptotic result.

The comparison between AIC-based and cross-validation-based tuning should be
interpreted as a finite-sample sensitivity analysis of two practically
available strategies for selecting \(\lambda\). Neither strategy uses knowledge
of the true active set or of a value of \(\lambda\) that would optimize
empirical coverage. Moreover, the oracle-type theory assumes consistent
recovery of the true active set but does not establish this property for either
of the two tuning strategies in full generality. The simulation results
therefore describe the finite-sample behavior of the proposed procedure under
practical tuning rather than under an oracle choice of \(\lambda\).
\end{newrevblock}

\rev{The choice of bootstrap interval also matters in finite samples.
Percentile and studentized intervals were generally more stable than
Wald-type bootstrap intervals and behaved similarly in several scenarios.
Percentile intervals use the estimated bootstrap distribution directly,
whereas studentized intervals additionally account for estimated scale. By
contrast, Wald-type intervals rely on normal critical values and are therefore
more sensitive to skewness and unstable information estimates.}

\rev{The debiased estimator provides an important high-dimensional benchmark,
but it addresses a different inferential problem. The simulations evaluate
coverage conditional on selection in the original sample, whereas debiased
procedures are generally designed for componentwise inference on a fixed
coordinate of the high-dimensional parameter. The relevant comparison is
therefore not whether either approach dominates uniformly, but whether the
proposed bootstrap provides coherent uncertainty quantification for effects
reported after \newrev{Cox--Lasso variable selection}. Remaining undercoverage in some
small-sample settings represents a finite-sample limitation of the present
implementation rather than a contradiction of the first-order result.}

The real data application illustrates that interval construction can
materially affect uncertainty statements for coefficients reported after
Cox--Lasso \newrev{variable selection}.
\newrev{The bootstrap-based intervals reflect the sampling variability of the
unpenalized Cox refit with the selected variable set held fixed; they do not
incorporate the additional variability that would arise from repeating the
selection step. They may still differ from naive Wald intervals because they
use the bootstrap distribution rather than a first-order normal approximation.}
Taken together, the theoretical, simulation, and application results indicate
that the model-based bootstrap provides a principled and computationally
feasible approach to this reporting problem.

The resampling construction is not intrinsically tied to
$\ell_1$-penalization, because it is based on a plug-in approximation to the
fitted multiplicative intensity structure rather than on a particular
property of the Lasso estimator. It may therefore also be useful for
finite-sample inference in standard Cox regression or after alternative
selection procedures, provided that appropriate stability and regularity
conditions hold. The same principle may be adaptable to related
semiparametric survival models, including competing risks and multistate
models.

Several limitations should be noted. The results establish first-order
validity only and do not imply higher-order refinements or second-order
accuracy. BCa intervals and related corrections are therefore not considered.
The analysis also assumes correct specification of the Cox model and focuses
on Lasso-based selection.
%\rev{The raw penalized Cox--Lasso coefficient is not treated as an
%inferential target because penalization introduces shrinkage bias and changes
%the limiting object. The results instead concern the unpenalized Cox estimator
%in the model selected from the original data.}
Extensions to other penalties, selection mechanisms, model misspecification,
and broader high-dimensional regimes require additional theory.
\plrev{A particularly relevant extension is a procedure-aware target that
incorporates the complete selection-and-estimation workflow rather than
conditioning on the model selected in the original sample.}
Further work may also examine higher-order refinements, support stability under
specific tuning rules, computational strategies for larger problems, and
extensions to competing risks and multistate settings.

\begin{acknowledgement}
\textit{Funding.} \rev{This work was supported by the Deutsche Forschungsgemeinschaft
(DFG, German Research Foundation; grant FR 439942859).}

\textit{Use of generative AI.} During the preparation of this work, the authors used ChatGPT solely to
improve the language and clarity of the manuscript. The authors reviewed and
edited all generated content as needed and take full responsibility for the
content of the manuscript.

\textit{Author contributions.} L.S. conceived the study, implemented the methods, performed the analyses, and
wrote the manuscript. S.F.-W. and A.G. contributed to the methodological development and
provided supervision, review, and critical feedback. All authors read and
approved the final manuscript.
\end{acknowledgement}
\vspace*{1pc}

\noindent {\bf Conflict of Interest}

\noindent The authors declare no conflicts of interest.
\vspace*{1pc}

\noindent {\bf Data availability}

\noindent \rev{The SEER data analyzed in this study are publicly available through the
National Cancer Institute SEER Program, subject to the SEER data-use agreement
\cite{seer_program}.}
\vspace*{1pc}

\noindent {\bf Ethics approval and patient consent}

\noindent \rev{This study used de-identified SEER registry data obtained under the SEER data-use agreement. No new patient data were collected and no intervention was performed; separate ethics approval and individual patient consent were therefore not required for this secondary analysis.}
\vspace*{1pc}

\noindent {\bf Code availability}

\noindent \rev{The complete code used for the simulation study and data analysis is
available at
\url{https://github.com/lena-222/BootstrapPartyLassoCox}.}
\vspace*{1pc}

\noindent {\bf Supporting Information}

\noindent \rev{Additional technical proofs, simulation details, supplementary tables, and supplementary figures are available as Supporting Information. Code for reproducing the simulation study and data analysis is available at
\url{https://github.com/lena-222/BootstrapPartyLassoCox}.}

% Precompiled bibliography for arXiv; avoids an additional BibTeX step.

\label{lastpage}

\end{document}